\documentclass[letterpaper, 10 pt, conference]{ieeeconf}  

\IEEEoverridecommandlockouts                              

\usepackage{amsfonts}
\usepackage{amstext}
\usepackage{mathtools}
\usepackage{cases}
\usepackage{enumerate}
\usepackage{subfigure}
\usepackage{amssymb}
\usepackage{float}
\usepackage{comment}
\usepackage[dvipsnames]{xcolor}
\usepackage{hyperref}

\newtheorem{thm}{\textbf{\text{Theorem}}}
\newtheorem{lem}{\textbf{\text{Lemma}}}
\newtheorem{pro}{\textbf{\text{Proposition}}}

\newtheorem{rmk}{\textbf{\text{Remark}}}

\newcommand{\ie}{\textit{i.e.}}
\newcommand{\eg}{\textit{e.g.}}

\title{\LARGE \bf Nonlinear Observer Design for Landmark-Inertial Simultaneous Localization and Mapping}
\author{Mouaad Boughellaba, Soulaimane Berkane, and Abdelhamid Tayebi
\thanks{This work was supported by the National Sciences and Engineering Research Council of Canada (NSERC), under the grants NSERC-DG RGPIN 2020-06270 and NSERC-DG RGPIN-2020-04759, and by Fonds de recherche du Qu\'ebec (FRQ).} 
\thanks{M. Boughellaba and A. Tayebi are with the Department of Electrical Engineering, Lakehead University, Thunder Bay, ON P7B 5E1, Canada \tt\small \{mboughel,atayebi\}@lakeheadu.ca.} 
\thanks{S. Berkane is with the Department of Computer Science and Engineering, University of Quebec in Outaouais, Gatineau, QC, Canada. {\tt\small Soulaimane.Berkane@uqo.ca}}
}

\begin{document}

\maketitle
\thispagestyle{empty}
\pagestyle{empty}

\begin{abstract}
This paper addresses the problem of Simultaneous Localization and Mapping (SLAM) for rigid body systems in three-dimensional space. 
We introduce a new matrix Lie group \(SE_{3+n}(3)\), whose elements are composed of the pose, gravity, linear velocity and landmark positions, and propose an almost globally asymptotically stable nonlinear geometric observer that integrates Inertial Measurement Unit (IMU) data with landmark measurements. The proposed observer estimates the pose and map up to a constant position and a constant rotation about the gravity direction.
Numerical simulations are provided to validate the performance and effectiveness of the proposed observer, demonstrating its potential for robust SLAM applications.
\end{abstract}

\section{Introduction}
The design of reliable autonomous robots (\eg, aerial vehicles, ground vehicles and marine vehicles) relies mainly on the development of efficient and reliable algorithms for self-localization, environment perception and motion control. In this context, a fundamental problem of great importance to the control and robotics community, namely the Simultaneous Localization And Mapping (SLAM), consists in estimating the robot's trajectory (\ie, position, linear velocity, and orientation) while simultaneously building a map of the unknown surrounding environment. Similar to vision-aided inertial navigation systems \cite{Wang_TAC2021}, visual-based SLAM can be complemented with inertial measurements, such as those provided by Inertial Measurement Units (IMUs), to improve the accuracy and robustness of SLAM algorithms. The use of visual sensors together with IMUs to solve SLAM problems is referred to as visual-IMU SLAM or visual-inertial SLAM \cite{Stefan_IJR2015,Campos_TR2021}.\\
Tracing back the literature on SLAM, one finds that the solutions to SLAM fall mainly into two categories: filter-based solutions and optimization-based solutions \cite{Strasdat_ICRA2010}. Filter-based solutions can be further classified into stochastic-based solutions and deterministic-based solutions. The vast majority of the proposed stochastic solutions are variants of the Kalman filter. The most popular variant is the Extended Kalman Filter (EKF) SLAM, which models the uncertainty in the system dynamics and the measurement model as additive Gaussian noise and also considers the position of the robot and landmarks as a single state vector \cite{SLAM_tutorial_1}. However, the dimension of the state vector defined in EKF SLAM grows continuously as new landmark observations are made. Consequently, EKF cannot effectively handle a large-scale SLAM due to its high computational complexity \cite{Alsayed_2017}. In addition, EKF SLAM suffers from a consistency problem that is well known in the literature \cite{Consistency_pb2,Consistency_pb1}. Therefore, a significant number of EKF-based SLAM solutions have been proposed to improve the scalability, robustness, and consistency of the standard EKF SLAM, such as compressed EKF SLAM \cite{compressed_EKF_SLAM}, unscented Kalman filter SLAM \cite{uncentered_Kalman_filte}, and information filter SLAM \cite{Information_Filter_SLAM}. Unfortunately, all EKF-based SLAMs guarantee only local stability due to the linearization in their framework. This limitation motivated the authors in \cite{Louren_2018,Louren_2016,Louren_2013,Guerreiro_TR2013} to propose sensor-based SLAM observers designed in sensor space to derive a linear time-varying system that mimics the nonlinear dynamics of the SLAM problem. As a result, global stability can be achieved by applying the Kalman filter approach to the derived linear time-varying system. However, this class of sensor-based observers does not include the robot pose in the dynamics of the SLAM problem.\\ 
Recently, many researchers have adopted nonlinear deterministic observers, designed directly on matrix Lie groups to achieve strong stability results for the SLAM problem. The work in \cite{invariant_EKF_SLAM} uses a symmetry-preserving approach \cite{Symmetry-preserving_observers}, along with invariant observer design techniques \cite{Bonnabel_2005} to propose a new variant of EKF SLAM, called invariant EKF SLAM (IEKF SLAM). Later in \cite{barrau_arxiv2016}, the authors introduced a new Lie group denoted by $SE_{1+n}(3)$ and proposed another IEKF SLAM on this Lie group. Although the IEKF SLAM outperforms the standard EKF SLAM in terms of consistency and robustness, it still relies on linearization to find the observer gain (\textit{i.e.,} no global stability results can be obtained). The group structure $SE_{k+1}(3)$ has been exploited also in \cite{Mahony_SLAM_CDC2017,Mahony_AJC2021} where geometric nonlinear SLAM observers are proposed. Other works, such as \cite{Miaomiao_SLAM_CDC_2018,Zlotnik_SLAM_TAC2018} have considered biased group velocity measurements. More recently, the authors in \cite{VANGOOR_aut_2021,Goor_SLAM_CDC2019} extended the geometric structure of the SLAM algorithm proposed in \cite{Mahony_SLAM_CDC2017} to include bearing vectors. The extended matrix Lie group structure $VSLAM_n(3)$ was used to design an almost semi-globally asymptotically stable nonlinear SLAM observer \cite{VANGOOR_aut_2021}.

Most of the aforementioned references assume the availability of the body-frame linear velocity of the robot. Unfortunately, this assumption is often not feasible in many low-cost applications, as measuring the body-frame linear velocity typically requires a sophisticated setup that can be prohibitively expensive. In the present paper, we propose a landmark-inertial SLAM observer that guarantees almost global asymptotic stability (AGAS) without the need for body-frame linear velocity measurements. Instead, our approach utilizes data from an IMU. The proposed observer estimates the robot's pose and the map up to an unknown constant position and rotation about the vertical axis (yaw), which represents the best result achievable given the inherent observability limitations of the SLAM problem (four unobservable directions, see \cite{martinelli2013observability}). Additionally, the observer is computationally efficient; its complexity increases linearly with the number of landmarks, making it more suitable for real-time applications compared to traditional Extended Kalman Filter (EKF)-based techniques, where complexity increases quadratically due to covariance updates.


\section{Preliminaries}\label{s2}
The sets of real numbers and the n-dimensional Euclidean space are denoted by $\mathbb{R}$ and $\mathbb{R}^n$, respectively. The set of unit vectors in $\mathbb{R}^n$ is defined as $\mathbb{S}^{n-1}:=\{x\in \mathbb{R}^n~|~x^\top x =1\}$. Given two matrices $A$,$B$ $\in \mathbb{R}^{m\times n}$, their Euclidean inner product is defined as $\langle \langle A,B \rangle \rangle=\text{tr}(A^\top B)$. The Euclidean norm of a vector $x \in \mathbb{R}^n$ is defined as $||x||=\sqrt{x^\top x}$, and the Frobenius norm of a matrix $A \in \mathbb{R}^{n\times n}$ is given by $||A||_F=\sqrt{\langle \langle A,A \rangle \rangle}$. The identity matrix is denoted by $I_n \in \mathbb{R}^{n \times n}$. The attitude of a rigid body is represented by a rotation matrix $R$ which belongs to the special orthogonal group $SO(3):= \{ R\in \mathbb{R}^{3\times 3} | \hspace{0.1cm}\text{det}(R)=1, R^\top R=I_3\}$. The tangent space of the compact manifold $SO(3)$ is given by $T_RSO(3):=\{R \hspace{0.1cm}\Omega \hspace{0.2cm} | \hspace{0.2cm} \Omega \in \mathfrak{so}(3)\}$, where $\mathfrak{so}(3):=\{ \Omega \in \mathbb{R}^{3\times 3} | \Omega^\top=-\Omega\}$ is the Lie algebra of the matrix Lie group $SO(3)$. The map $[.]_{\times}: \mathbb{R}^3 \rightarrow \mathfrak{so}(3)$ is defined such that $[x]_\times y=x \times y$, for any $x,y \in \mathbb{R}^3$, where $\times$ denotes the vector cross product on $\mathbb{R}^3$. The angle-axis parameterization of $SO(3)$, is given by $\mathcal{R}(\theta, v):=I_3+\sin\hspace{0.05cm}\theta \hspace{0.2cm}[v]^\times + (1-\cos\hspace{0.05cm}\theta)([v]^\times)^2$, where $v\in \mathbb{S}^2$ and  $\theta \in \mathbb{R}$ are the rotation axis and angle, respectively. The Kronecker product of two matrices $A$ and $B$ is denoted by $A \otimes B$. The map $\mathrm{vec}: \mathbb{R}^{n \times m}\rightarrow \mathbb{R}^{nm}$ is defined as $\mathrm{vec}([x_1~x_2~\hdots~x_m])=[x_1^\top, x_2^\top, \hdots, x_m^\top]^\top$ where $x_i \in \mathbb{R}^n$ for each $i \in \{1, 2, \hdots, m\}$. Inspired by \cite{barrau_arxiv2016}, we define the matrix Lie group, denoted by $SE_{3+n}(3) \in \mathbb{R}^{(6+n)\times(6+n)}$, as $SE_{3+n}(3):= \{ X=\mathcal{M}(R,x_1,x_2,x_3,\bold{x_L}):R\in SO(3), x_1, x_2, x_3\in \mathbb{R}^3, \bold{x_L} \in \mathbb{R}^{3 \times n}\}$, where the map $\mathcal{M}: SO(3)\times \mathbb{R}^3\times \mathbb{R}^3\times \mathbb{R}^3\times \mathbb{R}^{3\times n} \rightarrow \mathbb{R}^{(6+n)\times(6+n)}$ is defined as follows:

\vspace{-0.3cm}
{\small \begin{equation*} \mathcal{M}(R,x_1,x_2,x_3,\bold{x_L}):=\left[\begin{array}{r|r} \begin{matrix} R \qquad \end{matrix} & \begin{matrix} x_1&x_2&x_3&\bold{x_L}\end{matrix}\\ \hline 0_{(3+n)\times 3}& I_{3+n}\qquad \end{array}\right]. \end{equation*}}The inverse of $X$ is given by $X^{-1}=\mathcal{M}(R^\top, -R^\top x_1, -R^\top x_2, -R^\top x_3, -R^\top \bold{x_L}) \in SE_{3+n}(n)$. Moreover, for any $X_1, X_2 \in SE_{3+n}(3)$, one has $X_1 X_2 \in SE_{3+n}(3)$ and $X_1^{-1}X_1=X_1 X_1^{-1}=I_{6+n}$. The Lie algebra of  the matrix Lie group $SE_{3+n}(3)$, denoted by $\mathfrak{se}_{3+n}(3) \in \mathbb{R}^{(6+n)\times(6+n)}$, is given as follows:

\vspace{-0.3cm}
{\small\begin{align}
    \mathfrak{se}_{3+n}(3):= \{ &V=\mathcal{V}(\Omega,\xi_1,\xi_2,\xi_3,\boldsymbol{\xi_L}):\nonumber\\
    &\Omega\in \mathfrak{so}(3), \xi_1, \xi_2, \xi_3\in \mathbb{R}^3, \boldsymbol{\xi_L} \in \mathbb{R}^{3 \times n}\},\nonumber
\end{align}}where the map $\mathcal{V}: \mathfrak{so}(3)\times \mathbb{R}^3\times \mathbb{R}^3\times \mathbb{R}^3\times \mathbb{R}^{3\times n} \rightarrow \mathbb{R}^{(6+n)\times(6+n)}$ is defined as follows:

\vspace{-0.3cm}
{\small
\begin{equation*} \mathcal{V}(\Omega,\xi_1,\xi_2,\xi_3,\boldsymbol{\xi_L}):=\left[\begin{array}{r|r} \begin{matrix} \Omega \qquad \end{matrix} & \begin{matrix} \xi_1&\xi_2&\xi_3&\boldsymbol{\xi_L}\end{matrix}\\ \hline 0_{(3+n)\times 3}& 0_{3+n}\qquad \end{array}\right]. \end{equation*}}

\section{Problem Statement}\label{s3}
Let $\{\mathcal{I}\}$ and $\{\mathcal{B}\}$ be the inertial frame and the body-fixed frame attached to the center of mass of a rigid body, respectively. Consider the following dynamics of a rigid body and a set of $n$ static landmarks:

\vspace{-0.3cm}
{\small
\begin{align}
    \dot{R}&=R[\omega^\mathcal{B}]_\times \label{equ:dynamics1}\\
    \dot{p}&=v\label{equ:dynamics12}\\
    \dot{v}&=g+
    Ra^\mathcal{B}\label{equ:dynamics11}\\
    \dot{p}_i&=0,\label{equ:dynamics4}
\end{align}}where $R \in SO(3)$ is the orientation of frame $\{\mathcal{B}\}$ with respect to frame $\{\mathcal{I}\}$, $p \in \mathbb{R}^3$ and $v \in \mathbb{R}^3$ denote the position and the linear velocity of the rigid body expressed in the inertial frame $\{\mathcal{I}\}$, $p_i \in \mathbb{R}^3$ is the position of the $i$-th landmark expressed in $\{\mathcal{I}\}$, $g\in \mathbb{R}^3$ is the acceleration due to gravity expressed in $\{\mathcal{I}\}$, $a^\mathcal{B} \in \mathbb{R}^3$ is the apparent acceleration capturing all non-gravitational forces applied to the rigid body expressed in $\{\mathcal{B}\}$, $\omega^\mathcal{B}$ is the angular velocity of the rigid body expressed in $\{\mathcal{B}\}$. The system dynamics \eqref{equ:dynamics1}-\eqref{equ:dynamics4} can be captured by a state evolving on the Lie group $SE_{2+n}(3)$ since the gravity vector $g$ is constant and known. However, it is difficult to design an observer with global stability guarantees directly on this group due to the coupling resulting from the gravity direction when considering the group error. To remove this coupling, we extend the system with the additional (auxiliary) state $g$ which satisfies:

\vspace{-0.2cm}
{\small
\begin{equation}
    \dot g=0\label{equ:dynamics2}.
\end{equation}}The idea of introducing auxiliary states into the system to remove coupling in the error equation has been used for example in \cite{Wang_TAC2021} in the context of visual-inertial navigation with a known map. The resulting dynamics of \eqref{equ:dynamics1}-\eqref{equ:dynamics4} and \eqref{equ:dynamics2} can be captured by considering the extended state $X=\mathcal{M}(R,p,v,g,\bold{p_L})$, with $\bold{p_L}:=[p_1~p_2~\hdots~p_n]$, which evolves in the Lie group $SE_{3+n}(3)$. Consequently, in view of \eqref{equ:dynamics1}-\eqref{equ:dynamics4}, one has

\vspace{-0.3cm}
{\small
\begin{equation}\label{real_sys_on_group}
    \dot{X}=[X, H]+XV,
\end{equation}}with group velocity $V=\mathcal{V}\left([\omega^\mathcal{B}]_\times, 0_{3\times 1}, a^\mathcal{B}, 0_{3\times 1}, 0_{3\times n}\right)$ (obtained from IMU), 
where $[.,.]:\mathbb{R}^{(6+n)\times (6+n)} \times \mathbb{R}^{(6+n)\times (6+n)} \rightarrow \mathbb{R}^{(6+n)\times (6+n)}$ is the Lie bracket operator defined as $[X_1, X_2]=X_1 X_2-X_2 X_1$ for any $X_1, X_2 \in \mathbb{R}^{(6+n)\times (6+n)}$ and
\begin{equation*} H=\left[\begin{array}{r|r} \begin{matrix} 0_3 \qquad \end{matrix} & \begin{matrix} 0_{3\times (3+n)}\end{matrix}\\ \hline 0_{(3+n)\times 3}& \qquad S \qquad \end{array}\right] \end{equation*}
with ~~~{\small$S=\begin{bmatrix}
        0&0&0&\cdots&0&0\\
        1&0&0&\cdots&0&0\\
        0&1&0&\cdots&0&0\\
        0&0&0&\cdots&0&0\\
        \vdots&\vdots&\vdots&\ddots&\vdots&\vdots\\
          0&0&0&\cdots&0&0\\
    \end{bmatrix} \in \mathbb{R}^{(3+n)\times (3+n)}$}.\\
    We assume that the relative position between the landmarks and the rigid body is measured in the body-fixed frame $\{\mathcal{B}\}$ and given as follows: 

    \vspace{-0.3cm}
    {\small
\begin{equation}\label{equ:measurements}
    y_i=R^\top\left(p-p_i\right),
\end{equation}}where $i\in\{1, 2, \hdots, n\}$. 
Defining $\bold r_i := [0_{3\times 1}^\top~r_i^\top]^\top \in \mathbb{R}^{6+n}$ and $\bold y_i := [y_i^\top~r_i^\top]^\top \in \mathbb{R}^{6+n}$ where $r_i:=[-1~0~0~e_i^\top]^\top \in \mathbb{R}^{3+n}$ with $e_i$ denotes the $i$th basis vector of $\mathbb{R}^n$, one can rewrite the measurements \eqref{equ:measurements} as follows:
{\small
\begin{equation}\label{equ:measurements_on_group}
    \bold y_i=X^{-1} \bold r_i,
\end{equation}}for every $i \in \{1, 2, \hdots, n\}$. Given the dynamics \eqref{real_sys_on_group} and the measurements \eqref{equ:measurements_on_group}, our objective is to estimate the pose and the landmark positions simultaneously. This is actually a SLAM problem, which is the process of estimating the rigid body pose relative to the map $\{p_1, p_2, \hdots, p_n\}$ while simultaneously estimating the map. Since the SLAM problem is not observable, the best that one can achieve is an estimate of the rigid body pose and the map up to an unknown constant position and rotation about the vertical axis. In the sequel, we propose an almost globally asymptotically stable observer that estimates the rigid body extended pose (position, velocity, orientation) and the landmark positions up to an unknown constant position and rotation about the $z$-axis.
\section{Observer Design}
In view of \eqref{real_sys_on_group}, we propose the following observer structure on $SE_{3+n}(3)$ with a copy of the dynamics plus an innovation term:

\vspace{-0.2cm}
{\small
\begin{equation} \label{observer_sys_on_group}
    \dot{\hat X}=[\hat X, H]+\hat X V+\Delta \hat X,
\end{equation}}where $\hat X=\mathcal{M}(\hat R,\hat p,\hat v,\hat g,\bold{\hat p}_{\bold L})$ is the estimated state, with $\bold{\hat p}_{\bold L}:=[\hat p_1~\hat p_2 \hdots \hat p_n]$, and  $\Delta \in \mathfrak{se}_{3+n}(3)$ is an innovation term. Note that $\hat R \in SO(3)$, $\hat p \in \mathbb{R}^3$, $\hat v \in \mathbb{R}^3$, $\hat g \in \mathbb{R}^3$ and $\hat p_i \in \mathbb{R}^3$ denote the estimates of $R$, $p$, $v$, $g$ and $p_i$, for each $i\in\{1, 2, \hdots, n\}$, respectively. Let the innovation term for each landmark measurement be as follows:

\vspace{-0.2cm}
{\small
\begin{equation}
    \bold{\tilde y}_i= \bold r_i-\hat X \bold y_i,
\end{equation}}and define the $3\times n$ matrix $\tilde Y:=Q[\bold{\tilde y}_1\cdots \bold{\tilde y}_n]$ where $Q:=[I_3~0_{3 \times (3+n)}]$. The proposed innovation term for the pre-observer \eqref{observer_sys_on_group} is given by 
{\small
\begin{equation}
    \Delta=\mathcal{V}([\sigma]_\times,\tilde YK_p,\tilde YK_v,\tilde YK_g,\tilde Y\Gamma^\top)\in \mathfrak{se}_{3+n}(3),
\end{equation}}where $K_p, K_v, K_g\in\mathbb{R}^n$ and $\Gamma\in\mathbb{R}^{n\times n}$ are constant gains and $\sigma=k^R \left(\hat{g} \times g\right)$ for some $k_R>0$. In explicit form, the observer is written as follows:

\vspace{-0.3cm}
{\small
 \begin{align}
      \dot{\hat R}&=\hat R[\omega^\mathcal{B}+\hat R^\top\sigma]_\times\label{equ:obr_dynamics1}\\
      \dot{\hat p}&=[\sigma]_\times\hat p+\hat v+\sum_{j=1}^{n} k_j^p(\hat R y_j-\hat p+\hat p_j)\label{equ:obr_dynamics2}\\
    \dot{\hat v}&=[\sigma]_\times \hat v+\hat g+\hat R a^\mathcal{B}+\sum_{j=1}^{n} k_j^v(\hat R y_j-\hat p+\hat p_j)\label{equ:obr_dynamics3}\\
    \dot{\hat g}&=[\sigma]_\times\hat g+\sum_{j=1}^{n} k_j^g(\hat 
    R y_j-\hat p+\hat p_j)\label{equ:obr_dynamics4}\\
    \dot{\hat p}_i&=[\sigma]_\times\hat p_i+\sum_{j=1}^{n} \gamma_{ij}(\hat R y_j-\hat p+\hat p_j),\label{equ:obr_dynamics5}
\end{align}}where $k_i^*$ is the $i-$th element of $K_*$ and $\gamma_{ij}$ stands for the element-wise components of $\Gamma$. The structure of our proposed observer \eqref{equ:obr_dynamics1}-\eqref{equ:obr_dynamics5} is shown in Fig. \ref{first_obs}. Now let us proceed to the analysis of the closed-loop system. Using the fact that $\dot{\hat X}^{-1}=-\hat X^{-1} \dot{\hat X} \hat X^{-1}$ and defining the left-invariant geometric error $E=X \hat X^{-1}$, it follows from \eqref{real_sys_on_group} and \eqref{observer_sys_on_group} that

\vspace{-0.2cm}
{\small
\begin{equation}\label{E_on_group}
    \dot E = [E, H]-E\Delta.
\end{equation}}Note that $E=\mathcal{M}(\tilde R, \tilde p, \tilde v, \tilde g, \tilde{\bold{p}}_L)$ where $\tilde R:=R \hat R^T$, $\tilde p:=p-\tilde R \hat p$, $\tilde v:=v-\tilde R \hat v$, $\tilde g:=g-\tilde R \hat g$ and $\tilde{\bold{p}}_L:=[\tilde p_1~\tilde p_2~\hdots~\tilde p_n]$ with $\tilde p_i:=p_i-\tilde R \hat p_i$. It is also worth noting that the error dynamics are independent of both the system's trajectory and the input; a desirable property found in geometric observers on Lie groups, see for instance \cite{barrau2016invariant,Lageman_TAC2010}. Due to the unobservability of the considered SLAM problem, the stability of the closed-loop system is analysed using an \textit{ego-centric} error on the reduced Lie group $SE_{2+n}(3)$. First, consider the reduced state and reduced estimation state given by $X_r:=\mathcal{M}_r(R,v,g,\mathbf 1\otimes p-\mathbf{p}_L)$ and $\hat X_r:=\mathcal{M}_r(\hat R,\hat v,\hat g,\mathbf 1\otimes\hat p-\hat{\mathbf{p}}_L)$, respectively, where $\mathbf 1:=[1\cdots 1]\in\mathbb{R}^{1\times n}$ and the map $\mathcal{M}_r: SO(3)\times \mathbb{R}^3\times \mathbb{R}^3\times \mathbb{R}^{3\times n} \rightarrow \mathbb{R}^{(5+n)\times(5+n)}$ is defined as follows $\mathcal{M}_r(R,x_1,x_2,\bold{x_L}):=\left[\begin{array}{r|r} \begin{matrix} R \qquad \end{matrix} & \begin{matrix} x_1&x_2&\bold{x_L}\end{matrix}\\ \hline 0_{(2+n)\times 3}& I_{2+n}\qquad \end{array}\right]$. It is not difficult to show that the reduced error $E_r=X_r\hat X_r^{-1}$ obeys the following dynamics

{\small
\begin{equation}\label{reduced_E_on_group}
    \dot E_r = [E_r, H_r]-E_r\Delta_r,
\end{equation}}where 
{\small \begin{align*}
&\Delta_r=\mathcal{V}_r([\sigma]_\times,\tilde YK_v,\tilde YK_g,\tilde Y(K_p\otimes \mathbf 1-\Gamma^\top))\in \mathfrak{se}_{2+n}(3)\\
&H_r:=\left[\begin{array}{r|r} \begin{matrix} 0_3 \qquad \end{matrix} & \begin{matrix} 0_{3\times (2+n)}\end{matrix}\\ \hline 0_{(2+n)\times 3}& \qquad S_r \qquad \end{array}\right] \end{align*}}with ~~~{\small $S_r=\begin{bmatrix}
        0&0&1&\cdots&1&1\\
        1&0&0&\cdots&0&0\\
        0&0&0&\cdots&0&0\\
        0&0&0&\cdots&0&0\\
        \vdots&\vdots&\vdots&\ddots&\vdots&\vdots\\
          0&0&0&\cdots&0&0\\
    \end{bmatrix} \in \mathbb{R}^{(2+n)\times (2+n)}$.}\\
The map $\mathcal{V}_r: \mathfrak{so}(3)\times\mathbb{R}^3\times \mathbb{R}^3\times \mathbb{R}^{3\times n} \rightarrow \mathbb{R}^{(5+n)\times(5+n)}$ is defined as follows:
{\small \begin{equation*} \mathcal{V}_r(\Omega,\xi_1,\xi_2,\boldsymbol{\xi_L}):=\left[\begin{array}{r|r} \begin{matrix} \Omega \qquad \end{matrix} & \begin{matrix} \xi_1&\xi_2&\boldsymbol{\xi_L}\end{matrix}\\ \hline 0_{(2+n)\times 3}& 0_{2+n}\qquad \end{array}\right]. \end{equation*}}
\begin{figure}[t]
  \centering
\includegraphics[width=0.75\linewidth,height=3cm]{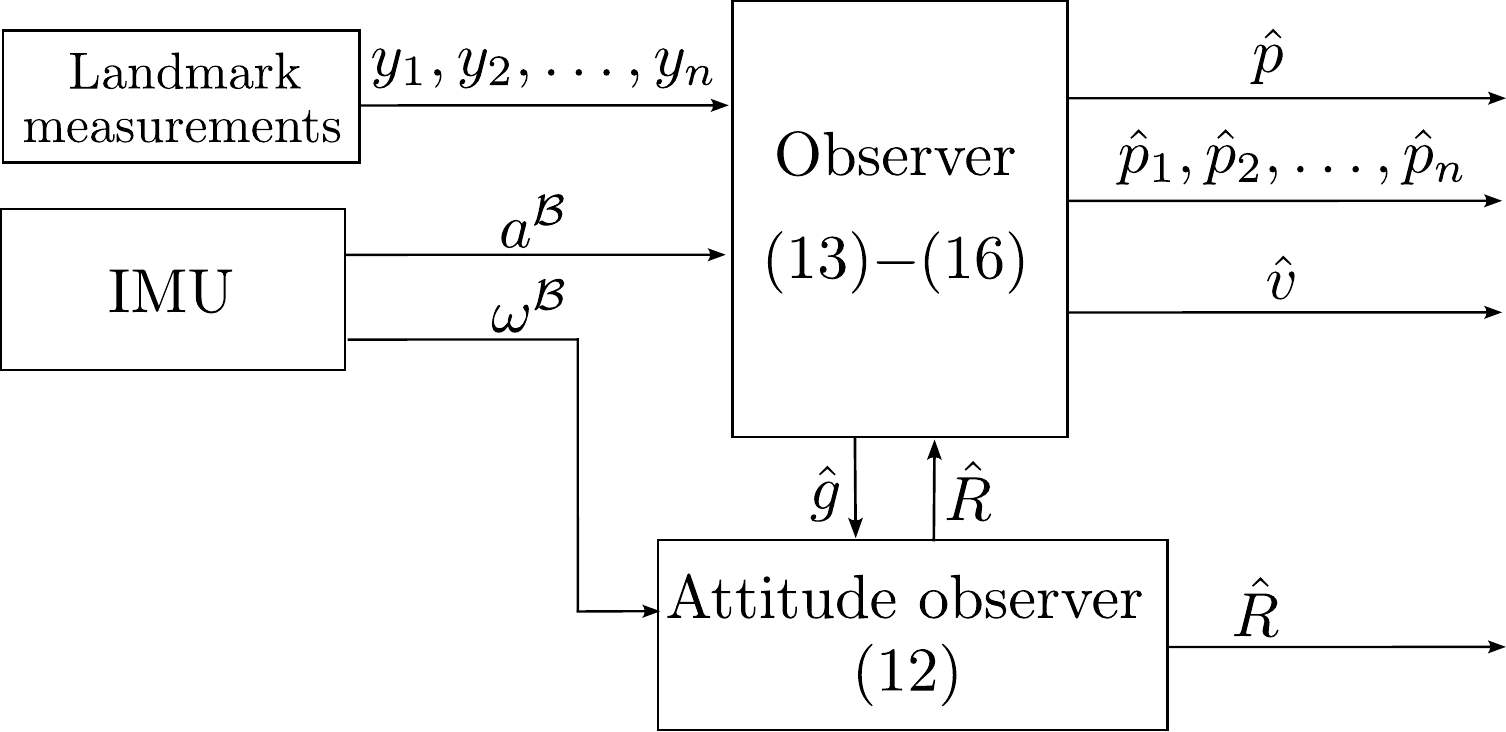}
  \caption{Structure of the proposed SLAM observer.}
  \label{first_obs}
\end{figure}
Letting $E_r=\mathcal{M}_r(\tilde R,\tilde v,\tilde g,\mathbf 1\otimes\tilde p-\tilde{\mathbf{p}}_L)$ and defining $x:=\left[\mathrm{vec}(\mathbf 1\otimes\tilde p-\tilde{\mathbf{p}}_L)^\top, \tilde v^\top, \tilde g^\top\right]^\top$, one can derive the following Linear Time Invariant (LTI) system
{\small \begin{equation}\label{lit_cls}
    \dot x = \left((A-LC)\otimes I_3\right)x,
\end{equation}}where {\small \begin{align}
    A&=\begin{bmatrix}
        0_{n \times n}&B_n\\
        0_{2\times n}&D\\
    \end{bmatrix},
    C=\begin{bmatrix}
        I_n&0_{n\times 2}
    \end{bmatrix},
    L=\begin{bmatrix}
       \mathbf 1^\top \otimes K_p^\top-\Gamma\\
        K_v^\top\\
        K_g^\top\\
    \end{bmatrix}\label{matrix:L}
\end{align}}
with {\small $B_n=\begin{bmatrix}
        \mathbf 1^\top & 0_{n \times 1}
    \end{bmatrix}\in \mathbb{R}^{n \times 2}~~\text{and}~~D=\begin{bmatrix}
        0&1\\
        0&0
    \end{bmatrix}$}.
\begin{lem}\label{lemma:lti_sys}
  The pair $(A,C)$ in \eqref{lit_cls} is Kalman observable, \ie, there exists a matrix gain $L$ such $(A-LC)$ is Hurwitz and the equilibrium point $x=0$ for the closed-loop system \eqref{lit_cls} is globally exponentially stable.
\end{lem}

\begin{proof}
    To prove this lemma, let us analyze the observability of the closed-loop system \eqref{lit_cls}. One can find the observability matrix of the system \eqref{lit_cls} to be 

    \vspace{-0.2cm}
{\small \begin{align}
    \mathcal{O}=\begin{bmatrix}
        I_n&0_{n\times 2}\\
        0_n&B_n\\
        0_n&\bar B_n\\
        0_n&0_{n\times 2}\\
        \vdots&\vdots\\
        0_n&0_{n\times 2}\\
    \end{bmatrix}, ~~\text{where}~~\bar B_n=\begin{bmatrix}
        0_{n \times 1} & \mathbf 1^\top
    \end{bmatrix}.
\end{align}}It follows from the fact that $\text{det}(\mathcal{O}^\top \mathcal{O})\neq0$ there exists a gain matrix $L$ such that the matrix $A-LC$ is Hurwitz. This guarantees the global exponential stability of the closed-loop system \eqref{lit_cls} at the equilibrium point $x=0$.
\end{proof}
Now let us analyse the dynamics of the rotation error $\tilde R$. In view of the dynamics of $E_r$ \eqref{reduced_E_on_group}, one has
{\small \begin{align}
    \dot{\tilde R}&=\tilde R\left[-k^R \left(\tilde R^\top g \times g\right)-\Pi(t)x\right]_\times,\label{R_tilde}
\end{align}}where $\Pi(t):=\left[0_3~0_3~\hdots~0_3~k^R g^\times \tilde R^\top\right] \in \mathbb{R}^{3 \times 3(n+2)}$. One checks that $||\Pi(t)||_F=\sqrt{2} k^R ||g||$. The closed-loop attitude dynamics \eqref{R_tilde} consists of two terms. The first term drives the estimated attitude to the actual one only up to a rotation about $g$ by aligning the vector $\tilde R^\top g$ with $g$, while the second term vanishes as the state $x$ converges to zero (according to Lemma \ref{lemma:lti_sys}) and the matrix $\Pi(t)$ remains bounded. Note that the proposed attitude observer estimates the attitude only up to a rotation about $g$ because the innovation term $\sigma$ relies only on a single vector $g$ \cite{Mahony_TAC2008}. Therefore, to simplify the stability analysis of \eqref{R_tilde}, we derive the following reduced attitude dynamics:

\vspace{-0.2cm}
{\small
\begin{align}
        \dot{\breve g}&=k^R \left[ \breve g \times g \right]_\times \breve g-[\breve g]_\times \Pi(t) x, \label{g_tilde}
    \end{align}}where $\breve g:= \tilde R^\top g$. The matrix $[\breve g]_\times \Pi(t)$ is bounded since $||\breve g||=||g||$ and $||\Pi(t)||_F=\sqrt{2} k^R ||g||$. Considering \eqref{lit_cls} and \eqref{g_tilde}, one obtains the following closed-loop system: 
    {\small \begin{align}
    \dot{\breve g}&=k^R \left[ \breve g \times g \right]_\times \breve g-[\breve g]_\times \Pi(t) x \label{csc_sys1}\\
    \dot x &= \left(\left(A-LC\right)\otimes I_3\right)x.\label{csc_sys2}
\end{align}}The above closed-loop system can be seen as a cascaded nonlinear system evolving on $\mathbb{S}^2_g \times \mathbb{R}^{3\left(n+2\right)}$ where $\mathbb{S}_g^2:=\{u \in \mathbb{R}^3:~u^\top u=||g||^2\}$.
Now, let us state the main result of this work.
\begin{thm}\label{theorem:thm1}
    Consider the closed-loop system \eqref{csc_sys1}-\eqref{csc_sys2}. Let $k^R>0$. Pick $L$ such that $(A-LC)$ is Hurwitz. Then, the equilibrium $(\breve g, x)=(g, 0)$ of the system \eqref{csc_sys1}-\eqref{csc_sys2} is AGAS.
\end{thm}
\begin{proof}
    See Appendix \ref{app_1}
\end{proof}
The AGAS result in Theorem \ref{theorem:thm1} is mainly due to the AGAS of the reduced attitude dynamics \eqref{csc_sys1}, which is the strongest stability result that can be obtained considering smooth time-invariant vector fields on the compact manifold $\mathbb{S}^2_g$, see for instance \cite{Koditschek}.

From the fact that $\breve g \rightarrow g$ (\ie, $R^\top g\rightarrow \hat R^\top g$), the rigid body orientation is estimated almost globally up to an unknown constant rotation about the vertical axis (gravity direction). In other words, $\tilde R\to R^*$ for some constant rotation $R^*$. Now, from $x\to 0$, one has $\tilde Y\to 0$ which implies that $\Delta\to 0$ and hence $\dot E\to [E,H]$. This further implies that $\dot{\tilde p}\to 0$ and, hence, $\tilde p=p-\tilde R\hat p=p-R^*\hat p$ converges to a constant (denoted $p^*$). The estimated position will converge to $(R^*)^\top (p-p^*)$.
Finally, from $\tilde Y\to 0$ and $\hat p\to (R^*)^\top(p-p^*)$, the landmark position estimates $\hat p_i$ converge to  $(R^*)^\top(p_i-p^*)$ and are therefore estimated up to the same constant pose. Finally, it should be noted that once the constant matrix gain $L$ is chosen to guarantee stability of the closed-loop matrix $(A-LC)$, the designer should select subsequently the observer gains $K_p,K_v,K_g$ and $\Gamma$. In view of \eqref{matrix:L}, the gains $K_v$ and $K_g$ are uniquely determined from the value of $L$ while the gains $K_p$ and $\Gamma$ should be chosen such as $\mathbf 1^\top \otimes K_p^\top-\Gamma=\begin{bmatrix}
    I_n&0&0
\end{bmatrix}L$. It follows that once $L$ is designed such that $(A-LC)$ is Hurwitz, one has multiple choices for $K_p$ and $\Gamma$. This is due to the fact that the SLAM system operates relative to the robot's own frame of reference (ego-centric). In essence, choosing $K_p$ and $\Gamma$ will influence
the estimated states, but does not influence the dynamics of the SLAM error $x$. 
A possible choice is to simply select $K_p=0$ and $\Gamma=-\begin{bmatrix}
    I_n&0&0
\end{bmatrix}L$. Another possible choice is to consider the `static environment' property and find the pair of gains $(K_p,\Gamma)$ that minimizes a criterion similar to \cite{VANGOOR_aut_2021}. A detailed analyzes of this direction, however, exceeds the scope of this work.

\section{SIMULATION}\label{s6}
In this section, we present some numerical simulations to illustrate the performance of the proposed observer. We consider a rigid body system moving along the trajectory $p(t)=3[\cos{t}~\sin{t}~1]^\top$, which represents a circular path with a diameter of 3 meters at a constant height of 3 meters. For the rotational motion, we assume that it is driven by an angular velocity $\omega^\mathcal{B}(t)=[-\cos{2t}~1~\sin{2t}]^\top$ with $R(0)=I_3$. Furthermore, we consider a set of fifteen landmarks that are randomly distributed in the space available for measurements. On the other hand, for the proposed observer, we consider the following initial conditions: $\hat R(0)=\mathcal{R}_\alpha(0.5\pi,u)$ with $u=[1~1~1]^\top$, $\hat p(0)=[0~0~0]^\top$, $\hat v(0)=[0~0~0]^\top$, $\hat g(0)=[0~0~0]^\top$, and $\hat p_i(0)=[0~0~0]^\top$ for every $i \in \{1, 2, \hdots, n\}$. We choose the following desired eigenvalues to design the matrix $L$: $\{-1~-2~-3~-4~-1~-2~-3~-4~-1~-2~-3~-4~-1~-2~-3~-4~-1\}$. For simplicity we choose $K_p=\mathbf 1^\top$. Then, given \eqref{matrix:L}, the gains $K_v, K_g$ and $\Gamma$ are chosen according to $L$ and $K_p$. We also choose $k_R=1$. The IMU and relative position measurements are simulated as follows: $\omega^\mathcal{B}_m=\omega^\mathcal{B}+\eta_\omega$, $a^{\mathcal{B}}=R^\top(\ddot{p}-g)+\eta_a$, and $y_i=R^\top\left(p-p_i\right)+\eta_{y_i}$  where $\eta_\omega$, $\eta_a$ and $\eta_{y_i}$ are zero-mean Gaussian noise signals with $0.01$, $0.2$ and $0.1$ variance, respectively. From the fact that the SLAM problem is not observable, $E$ is expected to converge to a constant value $E^*=\mathcal{M}(R^*, p^*, 0, 0, p^* \otimes \mathbf 1)$ instead of the identity, where $R^*=\tilde R(\infty)$ and $p^*=\tilde p(\infty)$. Therefore, for better visualization, we apply this constant transformation $(R^*, p^*)$ to correct the observer estimates back to the reference frame and the results are shown in Fig. \ref{trajectory} and Fig. \ref{map}.
\begin{figure}[H]
    \centering
    \includegraphics[width=0.8\linewidth]{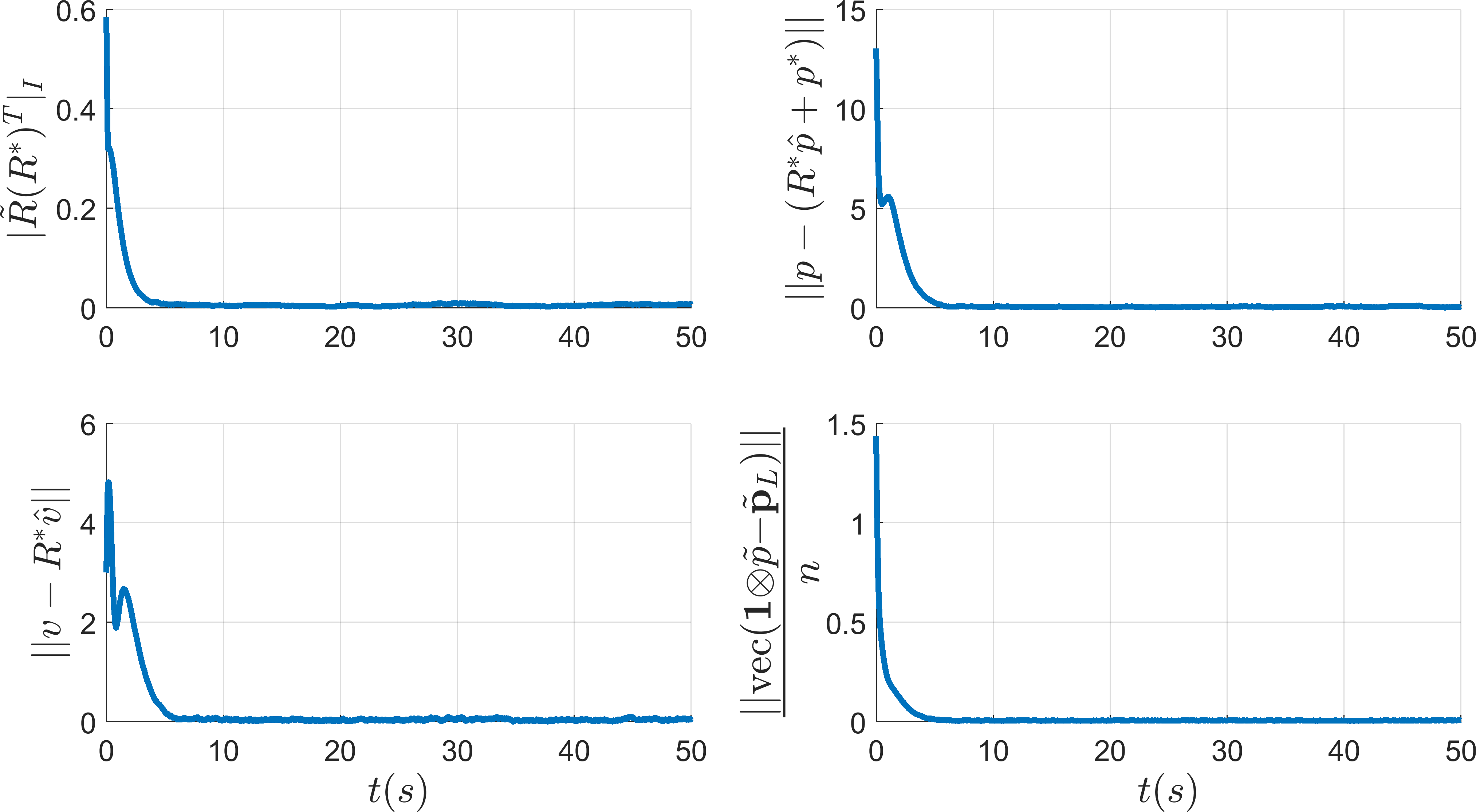}
    \caption{Estimation errors of rotation, position, velocity and landmarks considering the constant transformation $(R^*, p^*)$}
    \label{trajectory}
\end{figure}
\begin{figure}[H]
    \centering
    \includegraphics[width=0.8\linewidth]{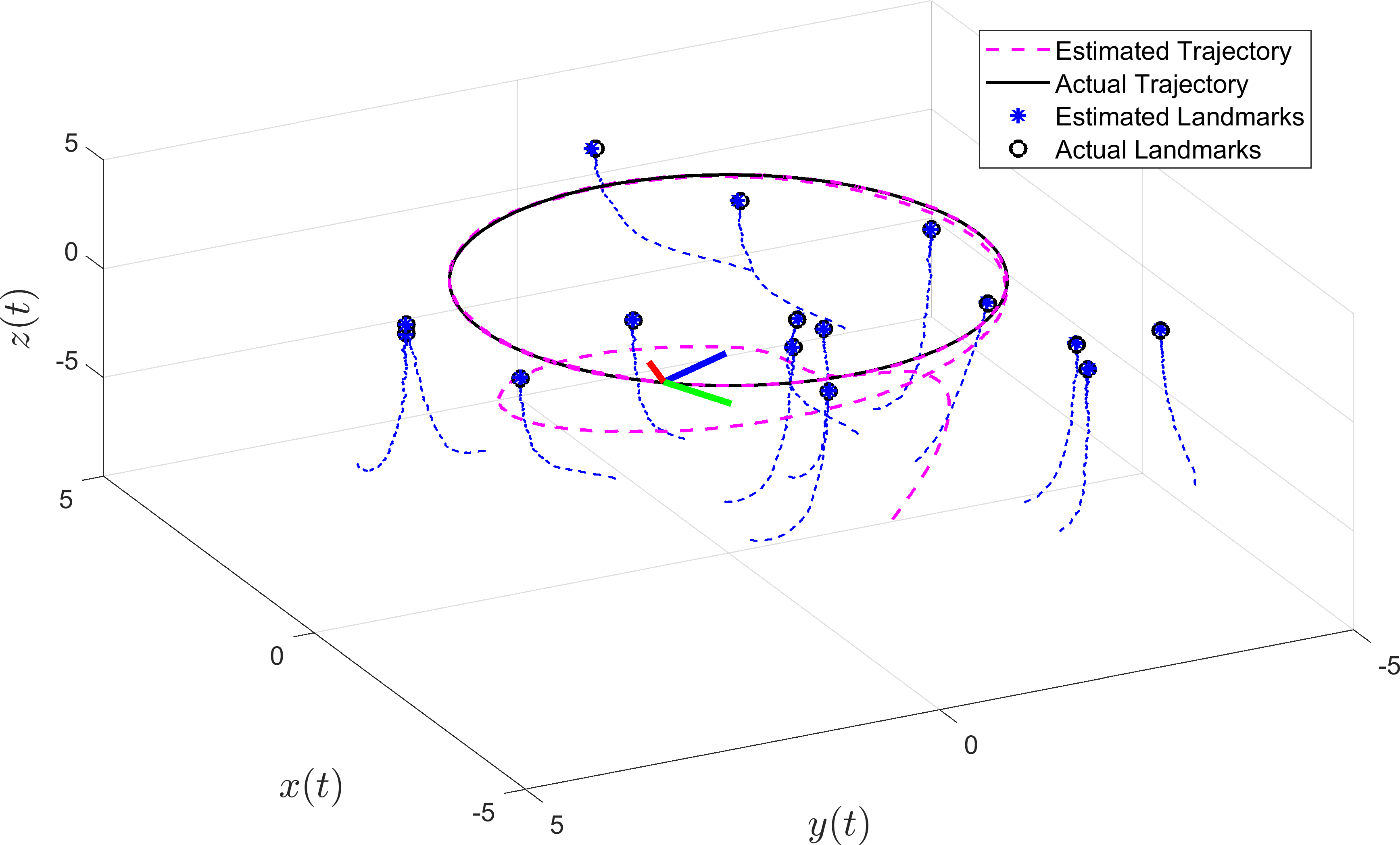}
    \caption{Estimated trajectories in the presence of noisy measurements and considering the constant transformation $(R^*, p^*)$. To avoid cluttering the figure, we have omitted the initial portions of the estimated landmark trajectories (blue dashed lines).}
    \label{map}
\end{figure}

\section{CONCLUSIONS}\label{s7}
In this paper, we proposed a new nonlinear geometric observer for the SLAM problem that leverages IMU data alongside landmark measurements. The observer effectively estimates (almost globally asymptotically) the robot's pose and map up to an unknown constant position and orientation about the gravity direction. One of the key advantages of our approach is its computational efficiency; the observer utilizes constant gains, resulting in a complexity that increases linearly with the number of landmarks. This stands in contrast to traditional Extended Kalman Filter (EKF)-based techniques, where complexity increases quadratically with the number of landmarks due to the covariance update process. Future work could explore several promising extensions, including the incorporation of biased measurements to enhance robustness against sensor noise and drift, as well as adapting the proposed observer to time-varying landmark positions that reflect dynamic environments. 




\appendices \label{app_1}
\section{Proof of Theorem 1} \label{app_1}
It follows from Lemma \ref{lemma:lti_sys} that the system \eqref{csc_sys2} is Globally Exponentially Stable (GES). Note also that system \eqref{csc_sys2} evolves independently of $\breve g$. Therefore, to determine the stability properties of the overall system \eqref{csc_sys1}-\eqref{csc_sys2}, we will first present two results considering the reduced attitude closed-loop system \eqref{csc_sys1}. The first result establishes the stability properties of the system \eqref{csc_sys1} with $x=0$, and the second studies the Input to State Stability (ISS) properties of \eqref{csc_sys1} with respect to $x$. Finally, we will conclude the stability properties of the overall system \eqref{csc_sys1}-\eqref{csc_sys2}. The following lemma represents the first result related to the system \eqref{csc_sys1} with $x=0$.
\begin{lem}\label{lem1}
    Let $k^R>0$. Then, the following statements hold:
     \begin{enumerate}[i)]
     \item System \eqref{csc_sys1}, with $x=0$, has two equilibria: the desired equilibrium $\breve g=g$ and the undesired one $\breve g=-g$.\label{eq_p}
     \item The desired equilibrium $\breve g=g$, for system \eqref{csc_sys1} with $x=0$, is almost globally asymptotically stable on $\mathbb{S}^2_g$.\label{des_eq}
     \end{enumerate}
\end{lem}
\begin{proof}
From \eqref{csc_sys1}, with $x=0$, it is clear that the equilibrium points on $\mathbb{S}^2_g$ are $\breve g=g$ and $\breve g=-g$. The time-derivative of the following positive definite function  $\mathcal{L}_1=\frac{1}{2}||g-\breve g||^2$, along the trajectories of \eqref{csc_sys1} with $x=0$, is given by  $\dot{\mathcal{L}}_1=-k^R||g \times \breve g||^2$. The time-derivative of the following positive definite function $\mathcal{L}_2=\frac{1}{2}||g+\breve g||^2$, along the trajectories of \eqref{csc_sys1} with $x=0$, is given by  $\dot{\mathcal{L}}_2=k^R||g \times \breve g||^2$. Therefore, the equilibrium point $\breve g=g$ is stable and the equilibrium point $\breve g=-g$ is unstable (repeller). Almost global asymptotic stability on $\mathbb{S}^2_g$ of $\breve g=g$ immediately follows. The time derivative of $\mathcal{L}_1$ and $\mathcal{L}_2$ were obtained using the fact that $u^\top [v]_\times u=0$ and $u^\top[u \times v]_\times v =-||u \times v||^2$ for every $u, v \in \mathbb{R}^3$. This completes the proof.
\end{proof}
Next, we will study the ISS property of \eqref{csc_sys1}, with respect to $x$ using the notion of almost global ISS introduced in \cite{Angeli_TAC2011}. 

\begin{lem}\label{lemma:ISS}
    Given $k^R>0$, the system \eqref{csc_sys1} is almost globally ISS with respect to the equilibrium point $\breve g =g$ and the input $x$.
\end{lem}

\begin{proof}
    From the fact that the system \eqref{csc_sys2} is globally exponentially stable, it follows that the state $x$ belongs to a compact set $\mathcal{A} \subset \mathbb{R}^{3(n+2)}$. This, together with the fact that $\breve g$ belongs to a compact manifold $\mathbb{S}_g^2$, ensures that the system \eqref{csc_sys1}, subject to the bounded inputs $x$, evolves on the compact manifold $\mathbb{S}_g^2\times\mathcal{A}$. Consequently, condition A0, given in \cite{Angeli_TAC2011}, is fulfilled. Moreover, considering the positive definite function  $\mathcal{L}_1$ and system \eqref{csc_sys1} with $x=0$, condition A1, given in \cite{Angeli_TAC2011}, is also fulfilled. To check condition A2, given in \cite{Angeli_TAC2011}, let us derive the first-order approximation of system \eqref{csc_sys1} with $x=0$. To do so, we apply small perturbation to $\breve g$ around the undesired equilibrium $-g$ through the attitude estimation error $\tilde R$ by letting $\tilde{R} = \mathcal{R}_\alpha(\pi, \bar{u}) \exp([\zeta]_\times)$, where $\bar{u} \perp g$ and $\zeta \in \mathbb{R}^3$ sufficiently small. Using the approximation $\exp([\zeta]_\times) \approx I_3 + [\zeta]_\times$ for sufficiently small $\zeta$, we have $\breve{g} = -g + [\zeta]_\times g$. Letting $z:= [\zeta]_\times g$, one can see $z$ as the small perturbation applied to $\breve g$ around the undesired equilibrium $-g$. By neglecting the cross terms and using the fact that $[u\times v]_\times =vu^\top-uv^\top$ for every $u, v \in \mathbb{R}^3$, one has $\dot z = k^R M z $, where $M:=\text{tr}(g g^\top) I_3-gg^\top=\text{diag}\left(||g||^2, ||g||^2, 0\right)$. This shows that the linearized dynamics of \eqref{csc_sys1} with $x=0$ have two identical positive eigenvalues $||g||$. Thus, condition A2 is also fulfilled. Now, consider the positive definite function  $\mathcal{L}_1$ whose time derivative along the trajectories \eqref{csc_sys1} is given by
    
    \vspace{-0.2cm}
    {\small\begin{align}
        \dot{\mathcal{L}}_1&=-k^R||g \times \breve g||^2-g^\top[\breve g]_\times \Pi(t) x\nonumber\\
        &=-k^R||[g]_\times \left(g-\bar g\right)||^2-g^\top[\breve g]_\times \Pi(t) x\nonumber\\
        &=-k^R ||g||^2||g-\bar g||^2-k^R\Big(||g||^2-2||g||^2g^\top\bar g\nonumber\\
        &~~~~~~~~~~~~~~~+\left(\bar g^\top g\right)^2\Big)-g^\top [\breve g]_\times \Pi(t) x. \label{equ_iss}
    \end{align}}The last equality was obtained using the fact that $[u]_\times^2=-||u||^2I_3+uu^\top$ for every $u \in \mathbb{R}^3$. Next, let $c_1$ and $c_2$ be a constant scalars such that $-k^R\big(||g||^2-2||g||^2g^\top\bar g+\left(\bar g^\top g\right)^2\big)\leq c_1$ and $||g^\top[\breve g]_\times \Pi(t)|| \leq c_2$. One has
    {\small \begin{align}
        \dot{\mathcal{L}}_1\leq-2k^R||g||^2\mathcal{L}_1+c_1+c_2||x||. \label{equ_iss}
    \end{align}}It follows from \eqref{equ_iss} that system \eqref{csc_sys1} satisfies the ultimate boundedness property introduced in \cite[Proposition 3]{Angeli_TAC2011}. Therefore, according to \cite[Proposition 2]{Angeli_TAC2011}, one can conclude that system \eqref{csc_sys1} is almost globally ISS with respect to $\breve g=g$ and the input $x$.
\end{proof}
Since the equilibrium $x=0$ for the system \eqref{csc_sys2} is GES and the system \eqref{csc_sys1} with $x=0$ is AGAS and almost globally ISS with respect to $x$, one can conclude that the cascaded system \eqref{csc_sys1}-\eqref{csc_sys2} is AGAS. This completes the
proof of Theorem.
\bibliographystyle{IEEEtran}
\bibliography{References}

\begin{thebibliography}{10}
\providecommand{\url}[1]{#1}
\csname url@samestyle\endcsname
\providecommand{\newblock}{\relax}
\providecommand{\bibinfo}[2]{#2}
\providecommand{\BIBentrySTDinterwordspacing}{\spaceskip=0pt\relax}
\providecommand{\BIBentryALTinterwordstretchfactor}{4}
\providecommand{\BIBentryALTinterwordspacing}{\spaceskip=\fontdimen2\font plus
\BIBentryALTinterwordstretchfactor\fontdimen3\font minus \fontdimen4\font\relax}
\providecommand{\BIBforeignlanguage}[2]{{%
\expandafter\ifx\csname l@#1\endcsname\relax
\typeout{** WARNING: IEEEtran.bst: No hyphenation pattern has been}%
\typeout{** loaded for the language `#1'. Using the pattern for}%
\typeout{** the default language instead.}%
\else
\language=\csname l@#1\endcsname
\fi
#2}}
\providecommand{\BIBdecl}{\relax}
\BIBdecl

\bibitem{Wang_TAC2021}
M.~Wang, S.~Berkane, and A.~Tayebi, ``Nonlinear observers design for vision-aided inertial navigation systems,'' \emph{IEEE Transactions on Automatic Control}, vol.~67, no.~4, pp. 1853--1868, 2022.

\bibitem{Stefan_IJR2015}
S.~Leutenegger, S.~Lynen, M.~Bosse, R.~Siegwart, and P.~Furgale, ``Keyframe-based visual–inertial odometry using nonlinear optimization,'' \emph{The International Journal of Robotics Research}, vol.~34, no.~3, pp. 314--334, 2015.

\bibitem{Campos_TR2021}
C.~Campos, R.~Elvira, J.~J.~G. Rodríguez, J.~M. M.~Montiel, and J.~D.~Tardós, ``Orb-slam3: An accurate open-source library for visual, visual–inertial, and multimap slam,'' \emph{IEEE Transactions on Robotics}, vol.~37, no.~6, pp. 1874--1890, 2021.

\bibitem{Strasdat_ICRA2010}
H.~Strasdat, J.~M.~M. Montiel, and A.~J. Davison, ``Real-time monocular slam: Why filter?'' in \emph{2010 IEEE International Conference on Robotics and Automation}, 2010, pp. 2657--2664.

\bibitem{SLAM_tutorial_1}
H.~Durrant-Whyte and T.~Bailey, ``Simultaneous localization and mapping: part i,'' \emph{IEEE Robotics Automation Magazine}, vol.~13, no.~2, pp. 99--110, 2006.

\bibitem{Alsayed_2017}
G.~Bresson, Z.~Alsayed, L.~Yu, and S.~Glaser, ``Simultaneous localization and mapping: A survey of current trends in autonomous driving,'' \emph{IEEE Transactions on Intelligent Vehicles}, vol.~2, no.~3, pp. 194--220, 2017.

\bibitem{Consistency_pb2}
S.~Huang and G.~Dissanayake, ``Convergence and consistency analysis for extended kalman filter based slam,'' \emph{IEEE Transactions on Robotics}, vol.~23, no.~5, pp. 1036--1049, 2007.

\bibitem{Consistency_pb1}
S.~Julier and J.~Uhlmann, ``A counter example to the theory of simultaneous localization and map building,'' in \emph{Proceedings 2001 ICRA. IEEE International Conference on Robotics and Automation}, vol.~4, 2001, pp. 4238--4243.

\bibitem{compressed_EKF_SLAM}
J.~Guivant and E.~Nebot, ``Optimization of the simultaneous localization and map-building algorithm for real-time implementation,'' \emph{IEEE Transactions on Robotics and Automation}, vol.~17, no.~3, pp. 242--257, 2001.

\bibitem{uncentered_Kalman_filte}
S.~J. Julier and J.~K. Uhlmann, ``{New extension of the Kalman filter to nonlinear systems},'' in \emph{Signal Processing, Sensor Fusion, and Target Recognition VI}, vol. 3068, 1997, pp. 182 -- 193.

\bibitem{Information_Filter_SLAM}
Y.~Liu and S.~Thrun, ``Results for outdoor-slam using sparse extended information filters,'' in \emph{2003 IEEE International Conference on Robotics and Automation}, vol.~1, 2003, pp. 1227--1233.

\bibitem{Louren_2018}
P.~Lourenço, P.~Batista, P.~Oliveira, and C.~Silvestre, ``A globally exponentially stable filter for bearing-only simultaneous localization and mapping with monocular vision,'' \emph{Robotics and Autonomous Systems}, vol. 100, pp. 61--77, 2018.

\bibitem{Louren_2016}
P.~Lourenço, B.~J. Guerreiro, P.~Batista, P.~Oliveira, and C.~Silvestre, ``Simultaneous localization and mapping for aerial vehicles: A 3-d sensor-based gas filter,'' \emph{Auton. Robots}, vol.~40, no.~5, p. 881–902, 2016.

\bibitem{Louren_2013}
P.~Louren\c{c}o, B.~J. Guerreiro, P.~Batista, P.~Oliveira, and C.~Silvestre, ``Preliminary results on globally asymptotically stable simultaneous localization and mapping in 3-d,'' in \emph{2013 American Control Conference}, 2013, pp. 3087--3092.

\bibitem{Guerreiro_TR2013}
B.~J.~N. Guerreiro, P.~Batista, C.~Silvestre, and P.~Oliveira, ``Globally asymptotically stable sensor-based simultaneous localization and mapping,'' \emph{IEEE Transactions on Robotics}, vol.~29, no.~6, pp. 1380--1395, 2013.

\bibitem{invariant_EKF_SLAM}
S.~Bonnabel, ``Symmetries in observer design: Review of some recent results and applications to ekf-based slam,'' in \emph{Robot Motion and Control 2011}.\hskip 1em plus 0.5em minus 0.4em\relax London: Springer London, 2012, pp. 3--15.

\bibitem{Symmetry-preserving_observers}
S.~Bonnabel, P.~Martin, and P.~Rouchon, ``Symmetry-preserving observers,'' \emph{IEEE Transactions on Automatic Control}, vol.~53, no.~11, pp. 2514--2526, 2008.

\bibitem{Bonnabel_2005}
S.~Bonnabel and P.~Rouchon, \emph{On Invariant Observers}.\hskip 1em plus 0.5em minus 0.4em\relax Springer Berlin Heidelberg, 2005, pp. 53--65.

\bibitem{barrau_arxiv2016}
A.~Barrau and S.~Bonnabel, ``An ekf-slam algorithm with consistency properties,'' \emph{arXiv:1510.06263}, 2016.

\bibitem{Mahony_SLAM_CDC2017}
R.~Mahony and T.~Hamel, ``A geometric nonlinear observer for simultaneous localisation and mapping,'' in \emph{2017 IEEE 56th Annual Conference on Decision and Control (CDC)}, 2017, pp. 2408--2415.

\bibitem{Mahony_AJC2021}
R.~Mahony, T.~Hamel, and J.~Trumpf, ``An homogeneous space geometry for simultaneous localisation and mapping,'' \emph{Annual Reviews in Control}, vol.~51, pp. 254--267, 2021.

\bibitem{Miaomiao_SLAM_CDC_2018}
M.~Wang and A.~Tayebi, ``Geometric nonlinear observer design for slam on a matrix lie group,'' in \emph{2018 IEEE Conference on Decision and Control (CDC)}, 2018, pp. 1488--1493.

\bibitem{Zlotnik_SLAM_TAC2018}
D.~E. Zlotnik and J.~R. Forbes, ``Gradient-based observer for simultaneous localization and mapping,'' \emph{IEEE Transactions on Automatic Control}, vol.~63, no.~12, pp. 4338--4344, 2018.

\bibitem{VANGOOR_aut_2021}
P.~{van Goor}, R.~Mahony, T.~Hamel, and J.~Trumpf, ``Constructive observer design for visual simultaneous localisation and mapping,'' \emph{Automatica}, vol. 132, p. 109803, 2021.

\bibitem{Goor_SLAM_CDC2019}
P.~v. Goor, R.~Mahony, T.~Hamel, and J.~Trumpf, ``A geometric observer design for visual localisation and mapping,'' in \emph{2019 IEEE 58th Conference on Decision and Control (CDC)}, 2019, pp. 2543--2549.

\bibitem{martinelli2013observability}
A.~Martinelli \emph{et~al.}, ``Observability properties and deterministic algorithms in visual-inertial structure from motion,'' \emph{Foundations and Trends{\textregistered} in Robotics}, vol.~3, no.~3, pp. 139--209, 2013.

\bibitem{barrau2016invariant}
A.~Barrau and S.~Bonnabel, ``The invariant extended kalman filter as a stable observer,'' \emph{IEEE Transactions on Automatic Control}, vol.~62, no.~4, pp. 1797--1812, 2016.

\bibitem{Lageman_TAC2010}
C.~Lageman, J.~Trumpf, and R.~Mahony, ``Gradient-like observers for invariant dynamics on a lie group,'' \emph{IEEE Transactions on Automatic Control}, vol.~55, no.~2, pp. 367--377, 2010.

\bibitem{Mahony_TAC2008}
R.~Mahony, T.~Hamel, and J.-M. Pflimlin, ``Nonlinear complementary filters on the special orthogonal group,'' \emph{IEEE Transactions on Automatic Control}, vol.~53, no.~5, pp. 1203--1218, 2008.

\bibitem{Koditschek}
D.~Koditschek, ``The application of total energy as a lyapunov function for mechanical control systems,'' \emph{Contemporary Mathematics, American Mathematical Society, 1989}, vol.~97, 02 1989.

\bibitem{Angeli_TAC2011}
D.~Angeli and L.~Praly, ``Stability robustness in the presence of exponentially unstable isolated equilibria,'' \emph{IEEE Transactions on Automatic Control}, vol.~56, no.~7, pp. 1582--1592, 2011.

\end{thebibliography}

\end{document}